\documentclass[12pt,reqno]{amsart}

\usepackage{amsmath}
\usepackage{amsfonts}
\usepackage{amssymb}
\usepackage{latexsym}   
\usepackage{graphicx}
\usepackage{color}
\usepackage{hyperref}
\usepackage{algorithm,algorithmic}

\usepackage{vmargin,fancyhdr}   

\newcommand{\beql}[1]{\begin{equation}\label{#1}}
\newcommand{\eeq}{\end{equation}}

\newcommand{\Mean}[1]{{\mathbb E}\left[{#1}\right]}
\newcommand{\Var}[1]{{\mathbb Var}\left[{#1}\right]}
\newcommand{\Cov}[1]{ {\mathbb Cov}\left[{#1}\right]}

\newcommand{\Prob}[1]{{{\bf{Pr}}\left[{#1}\right]}}

\newtheorem{lemma}{Lemma}
\newtheorem{theorem}{Theorem}

\newcommand{\hide}[1]{}

\title{Colorful Triangle Counting \\ and a \textsc{MapReduce} Implementation}

\author{Rasmus Pagh}
\address{IT University of Copenhagen\\
Rued Langgaards Vej 7, DK-2300 K{\o}benhavn S \\
Denmark} \email{pagh@itu.dk}
\author{Charalampos E. Tsourakakis}
\address{Department of Mathematical Sciences\\
Carnegie Mellon University\\
5000 Forbes Av., 15213\\
Pittsburgh, PA \\
U.S.A} \email{ctsourak@math.cmu.edu}


\begin{document}
\maketitle

\begin{abstract}

In this note we introduce a new randomized algorithm for counting triangles in graphs.  
We show that under mild conditions, the estimate of our algorithm is 
strongly concentrated around  the true number of triangles. 
Specifically,  if $p \geq \max{ ( \frac{\Delta \log{n}}{t}, \frac{\log{n}}{\sqrt{t}})}$, 
where $n$, $t$, $\Delta$ denote the number of vertices in $G$, the number of triangles in $G$,
the maximum number of triangles an edge of $G$ is contained, 
then for any constant $\epsilon>0$ our unbiased estimate $T$ is concentrated around  its expectation, i.e.,
$ \Prob{ |T - \Mean{T}| \geq \epsilon \Mean{T} }  = o(1)$.
Finally, we present a \textsc{MapReduce} implementation of our algorithm. 
\end{abstract}

\section{Introduction}

Triangle counting is a fundamental algorithmic problem with many applications. 
The interested reader is urged to see \cite{tsourakakis} and references therein. 
The fastest exact triangle counting algorithm to date (in terms of number of edges in the graph) is due to Alon, Yuster and Zwick \cite{alon} and runs in $O(m^{\frac{2\omega}{\omega+1}})$, where currently the matrix multiplication exponent $\omega$ is 2.371 \cite{CopperWino}.
For planar graphs linear time algorithms are known, e.g., \cite{pap:yan}.
Practical methods for exact triangle counting use instead enumeration techniques, see e.g., \cite{wagner} and references therein. 
For many applications, especially in the context of large social networks, an exact count is not crucial but rather a fast, high quality estimate. 
Most of the work on approximate triangle counting is sampling-based and has considered a (semi-)streaming setting \cite{yosseff,gionis:spam,buriol,jowhari,doulion}.
A different line of research is based on a linear algebraic approach  \cite{avron,icdm08}.
Currently to the best of our knowledge, the state-of-the-art approximate counting method relies on a hybrid 
algorithm that first sparsifies the graph and then samples 
triples according to a degree based partitioning trick \cite{tsourakakis}.

In this short note, we present a new sampling approach to approximating the number of triangles in a graph  $G(V,E)$, 
which significantly improves existing sampling approached. Furthermore, it is easily implemented in parallel. 
The key idea of our algorithm is to correlate the sampling of edges such that if two edges of a triangle are sampled, the third edge is always sampled. This decreases 
the degree of the multivariate polynomial that expresses the number of sampled triangles.
We analyze our method using a powerful theorem due to Hajnal and Szemer\'{e}di \cite{HajnalSzemeredi}. 
This note is organized as follows: in Section~\ref{prelim} we discuss the theoretical preliminaries
for our analysis and in Section~\ref{algo} we present our randomized algorithm. In Section~\ref{analysis}
we present our main theoretical results,  we analyze our algorithm and we
discuss some of its important properties. In Section~\ref{mapreduce} we present an implementation
of our algorithm in the popular \textsc{MapReduce} framework. Finally, in Section~\ref{concl} 
we conclude with future research directions.

\subsection{Algorithm} 
\label{algo} 

\begin{algorithm}[t]
\caption{Colorful Triangle Sampling} 
\begin{algorithmic}
\REQUIRE Unweighted graph $G([n],E)$
\REQUIRE Number of colors  $N=1/p$
\STATE Let $f: V \rightarrow [N]$ have uniformly random values
\STATE $E' \leftarrow \{ \{u,v\}\in E \; | \; f(u)=f(v) \}$
\STATE $T \leftarrow $ number of triangles in the graph $(V,E')$
\STATE {\bf return} $T/p^2$ 
\end{algorithmic}
\end{algorithm}

Our algorithm, summarized as Algorithm 1, samples each edge with probability~$p$, where $N=1/p$ is integer, as follows.
Let  $f: [n] \rightarrow [N]$ be a random coloring of the vertices of $G([n],E)$, such that for all $v \in [n]$ and $i \in [N]$, $\Prob{ f(v) = i } = p$.
We call an {\em edge monochromatic} if both its endpoints have the same color. 
Our algorithm samples exactly the set $E'$ of monochromatic edges, counts the number $T$ of triangles in $([n],E')$ (using any exact or approximate triangle counting algorithm), and multiplies this count by $p^{-2}$.

Previous work~\cite{tsourakakis_submitted,doulion} has used a related sampling idea, the difference being that edges were sampled {\em independently} with probability $p$.
Some intuition why this sampling procedure is less efficient than what we propose can be obtained by considering the case where a graph 
has $t$ edge-disjoint triangles. With independent edge sampling there will be no triangles left (with probability $1-o(1)$) if $p^3 t=o(1)$. 
Using our colorful sampling idea there will be $\omega(1)$ triangles in the sample with probability $1-o(1)$ as long as $p^2 t = \omega(1)$. 
This means that we can choose a smaller sample, and still get a accurate estimates from it.

\section{Theoretical Preliminaries}
\label{prelim}

In Section~\ref{conc} we make extensive use of the following version of the Chernoff bound \cite{chernoff}.

\begin{lemma}[Chernoff Inequality]
Let $X_1, X_2, \ldots, X_k$ be independently distributed $\{0,1\}$ variables with $E[X_i]=p$. Then for any $\epsilon > 0$, we have
$$\Prob{ |\frac{1}{k} \sum_{i=1}^k X_i - p| > \epsilon p } \leq 2 e^{-\epsilon^2pk/2}$$
\label{lem:chernoff}
\end{lemma}

Hajnal and Szemer\'{e}di \cite{HajnalSzemeredi} proved in 1970 the following conjecture of Paul Erd\"{o}s:

\begin{theorem}[Hajnal-Szemer\'{e}di Theorem]
\label{lem:hajnal}
Every graph with $n$ vertices and maximum vertex degree at most $k$
is $k+1$ colorable with all color classes of size at least $n/k$.
\end{theorem}

\section{Analysis} 
\label{analysis}

We wish to pick $p$ as small as possible but at the same time have a strong concentration of the estimate around its expected value. How
small can $p$ be? In Section~\ref{secondmoment} we present a second moment argument which gives a sufficient condition for picking $p$. 
Our main theoretical result, stated as Theorem~\ref{thrm:concentration} in Section~\ref{conc}, provides a sufficient condition to this question. 
In Section~\ref{running} we analyze the complexity of our method. 

\subsection{Second Moment Method} 
\label{secondmoment} 

Using the Second Moment Method we are able to obtain the following strong theoretical guarantee: 

\begin{theorem} 
Let $n$, $t$, $\Delta$, $T$ denote the number of vertices in $G$, the number of triangles in $G$,
the maximum number of triangles an edge of $G$ is contained and the number of monochromatic triangles
in the randomly colored graph respectively. Also let $N=\frac{1}{p}$ the number of colors used.
If $p \geq \max{ ( \frac{\Delta \log{n} }{t}, \frac{\log{n}}{\sqrt{t}})}$, then $T \sim \Mean{T}$ 
with probability 1-$o(1)$. 
\label{thrm:secondmoment} 
\end{theorem}

\begin{proof} 

By Chebyshev's inequality, if $\Var{T}=o(\Mean{T}^2)$ then $T \sim \Mean{T}$ with probability $1-o(1)$ \cite{alonspencer}.
Let $X_i$ be a random variable for the $i$-th triangle, $i=1,\ldots,t$, such that $X_i=1$
if the $i$-th triangle is monochromatic. 
The number of monochromatic triangles $T$ is equal to the sum of these indicator variables, i.e., $T = \sum_{i=1}^t X_i$. By the linearity of expectation and by the fact that $\Prob{X_i=1} = p^2$ we obtain that $\Mean{T}=p^2t$. 
We set $\Delta = \sum_{i \sim j} \Prob{X_i \wedge X_j}$ where the sum is over ordered pairs 
and $i \sim j$ denotes that the corresponding indicator variables are dependent.
It is easy to check that the only case where two indicator variables are dependent
is when they share an edge. In this case the covariance is non-zero and for any $p>0$,
$ \Cov{X_i,X_j} = p^3-p^4 < p^3$.

Hence, we obtain the following upper bound on the variance of $T$, where $\delta_e$ 
is the number of triangles edge $e$ is contained and $\Delta =\max_{e\in E(G)} \delta_e$: 

\begin{align*}
\Var{T} \leq \Mean{T} + \Delta \leq p^2t + p^3 \sum_e \delta_e^2 \leq p^2t + 3p^3 t \Delta 
\end{align*}

We pick $p$ large enough to get $\Var{X}=o(\Mean{X}^2)$. It suffices:

\beql{eqsec}
p^4t^2 >> p^2t + 3 p^3 t \Delta  \rightarrow  p^2 t  >> 1 + 3 p\Delta 
\eeq

We consider two cases:

\noindent
\underline{$\bullet$ {\sc Case 1} ($p\Delta < 1/3$):}\\
It suffices that $p^2t = \omega(n)$ where $\omega(n)$ is some slowly growing function.
We pick $\omega(n)=\log^2{n}$ and hence $p \geq \frac{\log{n}}{\sqrt{t}}$. 

\noindent \\
\underline{$\bullet$ {\sc Case 2} ($p\Delta \ge 1/3$):}\\
It suffices that $\frac{pt}{\Delta} = \log{n}$. 

Combining the above two cases we get that if 

$$p \geq \max{ ( \frac{\Delta\log{n} }{t}, \frac{\log{n}}{\sqrt{t}})}$$

Equation~\ref{eqsec} is satisfied and hence $X \sim \Mean{X}$ with probability $1-o(1)$. 

\end{proof}

\subsection{Concentration via the Hajnal-Szemer\'{e}di Theorem} 
\label{conc}

Here, we present a different approach to obtaining concentration, based on partitioning 
the set of triangles/indicator variables in sets containing many independent random indicator variables and then taking 
a union bound. Our theoretical result is the following theorem:

\begin{theorem} 

Let $t_{\max}$ be the  maximum number of triangles a vertex $v$ is contained in. Also, let $n,t,p,T$ be defined as above, $\epsilon$ a small positive constant and $d > 0$ any constant. If $p^2 \geq \frac{4(d+3) t_{\max}  \log{n}}{\epsilon^2 t} $, then $\Prob{ |T - \Mean{T}| > \epsilon \Mean{T} } \leq \frac{1}{n^d}$.

\label{thrm:concentration} 
\end{theorem}

\begin{proof} 

Let $X_i$ be defined as above, $i=1,\ldots,t$. Construct an auxiliary graph $H$ as follows: add a vertex in $H$ for every triangle in $G$ and connect two vertices representing triangles $t_1$ and $t_2$ if and only if they have a common vertex. The maximum degree of $H$ is  $3t_{\max}=O(\delta^2)$, where $\delta=O(n)$ is the maximum degree in the graph. 
Invoke the Hajnal-Szemer\'{e}di Theorem on $H$: we can partition the vertices of $H$ (triangles of $G$) into sets $S_1, \dots, S_q$ such that $|S_i| > \Omega(\tfrac{t}{t_{\max}})$ and $q=\Theta(t_{\max})$. Let $k=\tfrac{t}{t_{\max}}$. Note that the set of indicator variables $X_i$
corresponding to any set $S_j$ is independent. Applying the Chernoff bound for each set $S_i, i=1,\ldots,q$  we obtain 

$$Pr \left[ |\frac{1}{k} \sum_{i=1}^k X_i - p^2| > \epsilon p^2 \right] \leq 2e^{-\epsilon^2p^2k/2}$$.

If $p^2 k \epsilon^2 \geq 4d'\log{n}$, then $2e^{-\epsilon^2p^2k/2}$ is upper bounded by $n^{-d'}$,
where $d'>0$ is a constant. Since $q=O(n^3)$ by taking a union bound over all sets $S_i$ we see that the triangle count is approximated
within a factor of $\epsilon$ with probability at least $1 - n^{3-d'}$ Setting $d=d'-3$ completes the proof. 
\end{proof}

\subsection{Complexity}
\label{running}

The running time of our procedure of course depends on the subroutine we use on the second step, i.e., to count triangles in the edge set $E'$.
Assuming we use an exact method that examines each vertex independently and counts the number of edges among its neighbors (a.k.a.
Node Iterator method \cite{wagner}) our algorithm runs in $O(n+m+p^2 \sum_{i \in [n]} \text{deg}(i))$ expected time
\footnote{We assume that uniform sampling of a color takes constant time. If not, then we obtain the term $O(n\log{(\frac{1}{p})}$
for the vertex coloring procedure.} by efficiently storing the graph and retrieving the neighbors
of $v$ colored with the same color as $v$ in $O(1+p\,\text{deg}(v))$ expected time. Note that this implies that the speedup with respect to the counting task is $1/p^2$.

\subsection{Discussion}
The use of Hajnal-Szemer\'{e}di Theorem in the context of proving concentration is not new, e.g., \cite{janson,tsourakakis}.
Despite the fact that the second moment argument gave us strong conditions on $p$, 
the use of Hajnal-Szemer\'{e}di has the potential of improving the $\Delta$ factor. 
The condition we provide on $p$ is {\em sufficient} to obtain concentration. Note --see Figure~\ref{fig:fig1}-- that it was necessary to partition the triangles into vertex disjoint  rather than edge disjoint triangles since we need mutually independent variables per chromatic class in order to apply the Chernoff bound. Were we able to remove the dependencies 
in the chromatic classes defined by edge disjoint triangles, probably the overall result could be improved. 
It's worth noting that for $p=1$ we obtain that $t \geq n \omega(n)$, where $\omega(n)$ is any slowly growing function of $n$. 
This is --to the best of our knowledge-- the mildest condition on the triangle density needed for a randomized algorithm to obtain concentration.

Furthermore, the powerful theorem of Kim and Vu \cite{kim-vu,vu} that was used in previous work \cite{tsourakakis_submitted} is not immediately applicable here: let  $Y_e$ be an indicator variable for each edge $e$ such that $Y_e=1$ {\it if and only if} $e$ is monochromatic, i.e., both its endpoints receive the same color. Note that the number of triangles is a boolean polynomial 
$T = \frac{1}{3} \sum_{\Delta(e,f,g)} \big( Y_e Y_f+ Y_f Y_g+Y_e Y_g \big)$ but the boolean variables are not independent
as the Kim-Vu \cite{kim-vu} theorem requires. It's worth noting that the degree of the polynomial
is two. Essentially, this is the reason for which our method obtains better results than existing work \cite{tsourakakis_submitted}
where the degree of the multivariate polynomial is three \cite{993069,tsourakakis_submitted}. 
It's worth noting that previous work \cite{tsourakakis,tsourakakis_submitted} sampled edges independently 
whereas our new method samples subsets of vertices but in a careful manner in order to decrease the degree
of the multivariate polynomial. 
Finally, it's worth noting that using a simple doubling procedure \cite{tsourakakis_submitted} and the median boosting trick of Jerrum, Valiant and Vazirani \cite{jerrum}
we can pick $p$ effectively in practice despite the fact that it depends on the quantity $t$ which we want to estimate by introducing an extra logarithm in the running time.

\begin{figure*} 
\includegraphics[width=0.4\textwidth]{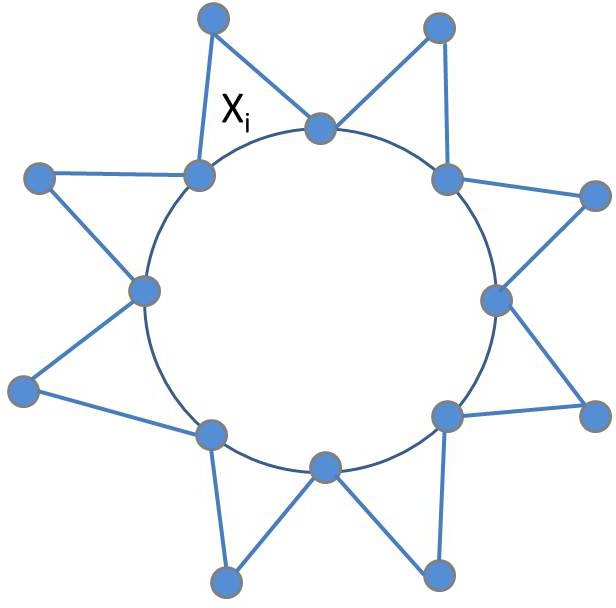}
\caption{Consider the indicator variable $X_i$ corresponding to the $i$-th triangle. Note that $\Prob{X_i|\text{rest are monochromatic}}=p \neq \Prob{X_i}=p^2$.
The indicator variables are {\em pairwise} but not mutually independent. }
\label{fig:fig1} 
\end{figure*}

Finally, from an experimentation point of view, it's interesting to see how well the upper bound $3\Delta t$ matches the sum $\sum_{e \in E(G)} \delta_e^2$
and the typical values for $\Delta$ and $t_{\max}$ in real-world networks. The following table shows these numbers for five networks
\footnote{AS:Autonomous Systems, Oregon: Oregon route views, Enron: Email communication network, ca-HepPh and AstroPh:Collaboration networks. Self-edges were removed.} taken from the SNAP library \cite{snap}.
We see that $\Delta$ and $t_{\max}$ are significantly less than their upperbounds and that typically $3\Delta t$ is significantly larger than 
$\sum_{e \in E(G)} \delta_e^2$ except for the collaboration network of Arxiv Astro Physics. The results are shown in Table~\ref{tab:datasets}.

\begin{table}[ht]
\begin{center}
\begin{tabular}{|l|r|r|r|r|r|r|r|} \hline
Name    & Nodes        & Edges   &  Triangle Count & $\Delta$ & $t_{\max}$& $\sum_{e \in E(G)} \delta_e^2$ & 3$\Delta t$        \\ \hline 
AS       & 7,716        & 12,572  &  6,584          &   344    & 2,047     &        595,632                 & 6,794,688           \\ \hline
Oregon   & 11,492       & 23,409  &  19,894         &   537    & 3,638     &       2,347,560                & 32,049,234         \\ \hline
Enron    & 36,692       & 183,831 &  727,044        &   420    & 17,744    &      75,237,684                & 916,075,440          \\ \hline
ca-HepPh & 12,008       & 118,489 & 3,358,499       &   450    & 39,633    &    1.8839 $\times 10^9$        & 4.534$\times 10^9$  \\ \hline
AstroPh  & 18,772       & 198,050 & 1,351,441       &   350    & 11,269    &      148,765,753               & 1.419$\times 10^9$  \\ \hline
\end{tabular}
\end{center}
\caption{Values for the variables involved in our formulae for five real-world networks.}
\label{tab:datasets}
\end{table}

\section{A \textsc{MapReduce} Implementation} 
\label{mapreduce} 

\textsc{MapReduce} \cite{dean} has become the {\it de facto} standard in academia and industry
for analyzing large scale networks.  Recent work by Suri and Vassilvitskii \cite{suri} proposes two algorithms for counting triangles. 
The first is an efficient \textsc{MapReduce} implementation of the Node Iterator algorithm, see also \cite{wagner}
and the second is based on partitioning the graph into overlapping subsets so that each triangle is present in at
least one of the subsets. 

Our method is amenable to being implemented in \textsc{MapReduce}
and the skeleton of such an implementation is shown in Algorithm 2\footnote{It's worth pointing out for completeness reasons 
that in practice one would not scale the triangles
after the first reduce. It would emit the count of monochromatic triangles
which would be summed up in a second round and scaled by $1/p^2$.}.
We implicitly assume that in a first round vertices have received 
a color uniformly at random from the $N$ available colors and that we have the coloring
information for the endpoints of each edge. 
Each mapper receives an edge together with the colors of its edgepoints. 
If the edge is monochromatic, then it's emitted with the color as the key and the edge as the value. 
Edges with the same color are shipped to the same reducer where locally a triangle counting
algorithm is applied. The total count is scaled appropriately. Trivially, the following lemma  
holds by the linearity of expectation and the fact that the endpoints of any edge 
receive a {\em given} color $c$ with probability $p^2$.

\begin{algorithm}[t]
\caption{\textsc{MapReduce} Colorful Triangle Counting $G(V,E),p=1/N$ } 
\begin{algorithmic}
\STATE {\bf Map:} Input $\langle e=(u,f(u),v,f(v));1 \rangle$
\COMMENT{Let $f$ be a uniformly at random coloring of the vertices with $N$ colors}
\STATE  {\it if} $f(u)=f(v)$ {\it then} emit $\langle f(u); (u,v)\rangle$
\STATE {\bf Reduce:} Input $\langle c; E_c=\{ (u,v) \} \subseteq E \rangle$ 
\COMMENT{ Every edge $(u,v) \in E_c$ has color $c$, i.e., $f(u)=f(v)$}
\STATE Scale each triangle by $\frac{1}{p^2}$. 
\end{algorithmic}
\end{algorithm}

\begin{lemma} 
The expected size to any reduce instance is $O(p^2m)$ and the expected total space used at the end of the map
phase is $O(pm)$. 
\label{mrc} 
\end{lemma}

\section{Conclusions}
\label{concl} 

In this note we introduced a new randomized algorithm for approximate triangle counting,
which is implemented easily in parallel. We showed such an implementation in the popular \textsc{MapReduce}
programming framework.  
The key idea which improves the existing work is that by our new sampling method the degree of the multivariate 
polynomial expressing the number of triangles decreases by one, compared to previous work, e.g., \cite{993069,tsourakakis_submitted}.
We used the powerful result of Hajnal-Szemer\'{e}di Theorem to obtain a concentration result
which is unlikely to be the best possible. 
We observe that our result extends any subset of triangles satisfying some predicate (e.g., containing a certain vertex), 
in the sense that counting such triangles in the sample leads to a concentrated estimate of the number in the original graph.

In future work we plan to investigate sampling methods for counting triangles in weighted graphs, other types of subgraphs and
several systems-oriented aspects of our work. 

\section{Acknowledgments} 

The authors are pleased to acknowledge the valuable feedback of Tom Bohman and Alan Frieze.

\end{document}